\PassOptionsToPackage{table,dvipsnames}{xcolor}
\documentclass[nonacm,onecolumn,oneside,manuscript]{acmart}
\pdfoutput=1 

\usepackage{pbalance}
\usepackage{numprint}

\usepackage{latexsym}
\usepackage{booktabs}
\usepackage{enumitem}
\usepackage{graphicx}
\usepackage{color}

\usepackage{mathtools}
\usepackage{bm}
\usepackage{verbatim}
\usepackage{subcaption}
\usepackage{todonotes}
\usepackage{hyperref}
\usepackage[]{algpseudocode}
\usepackage{algorithm}

\newtheorem{theorem}{Theorem}
\newtheorem{lemma}{Lemma}

\usepackage{physics}
\usepackage{tikz}
\usepackage{mathdots}
\usepackage{cancel}
\usepackage{color}
\usepackage{siunitx}
\usepackage{array}
\usepackage{multirow}
\usepackage{gensymb}
\usepackage{tabularx}
\usepackage{extarrows}
\usepackage{booktabs}
\usetikzlibrary{fadings}
\usetikzlibrary{patterns}
\usetikzlibrary{shadows.blur}
\usetikzlibrary{shapes}

\title{Hamm-Grams: An Algorithm for Mining Regular Expressions of Bytes}

\author{Derek Everett}
\email{derekeverett@gmail.com}
\affiliation{%
  \institution{Amazon}
  \city{Baltimore}
  \country{USA} 
  }

\author{Edward Raff}
\email{edward.raff@crowdstrike.com}
\affiliation{%
  \institution{CrowdStrike}
  \city{Austin}
  \country{USA} 
  }

\author{James Holt}
\email{james@holt.net}
\affiliation{%
  \institution{CrowdStrike}
  \city{Austin}
  \country{USA} 
  }

\begin{document}

\begin{abstract}

Malware poses a critical and ever-evolving threat, and robust and effective systems for detecting and classifying malware are of essential importance. 
$n$-grams features are among the common static features used in effective machine learning systems for malware, but these features are inherently brittle.
We propose an algorithm for constructing more robust features, hamm-grams, which are a special class of regular expressions having a fixed length and single-character wildcards. 
We devise an efficient algorithm for finding common hamm-grams using a new locality-sensitive hash designed to produce collisions among pairs of small Hamming distance and a clustering within hash buckets to place wildcards. 
We then demonstrate the advantages of these features in malware classification and detection tasks. 
\end{abstract}

\maketitle

\section{Introduction}

The detection of malware is a problem of utmost importance. 
Therefore, designing reliable and robust machine-learning features for malware detection is an ongoing field of research and activity.  
The need for cheap deployment and training costs is of particular relevance, as malware is increasingly observed in extremely limited compute environments 
with resources too limited to run parsing. 
Notable incidents have included
industrial control systems~\citep{5772960} and satellites~\citep{VANCAMP2022101458} as targets of malware. 
Such systems may have as little as 256 KB of system RAM~\citep{lin2020mcunet}, requiring especially small models.

$n$-gram features, sequences of $n$ consecutive characters, form a powerful tool in the machine learning toolbox. 
Combined with linear classifiers such as $L_1$-penalized logistic regression, $n$-grams create a formidable and effective baseline and are commonly used in practice. Such models are cheap to train, do not require GPU co-processors, and are small and fast enough to deploy in almost all situations. 
For static malware detection, $n$-grams of bytes are commonly used features;
however, as $n$ increases, the likelihood of an $n$-gram re-occurring decreases, as all $n$ tokens from the alphabet $\Sigma$ must match. 
If any single entry changes, the feature effectively ``disappears'', making $n$-grams brittle when larger values of $n$ are desired.

An example that illustrates one failure case of $n$-grams for malware detection is the following x86 assembly sequence:
\begin{verbatim}
push ebp
sub edx, edx
mov ebp, esp
\end{verbatim}
which corresponds to the byte 5-gram (in hex)
\texttt{55 29 D2 89 E5}. The second assembly instruction sets the value to zero and could be equivalently replaced with the instruction
\texttt{xor edx,edx},
which would change the 5-gram to \texttt{55 31 D2 89 E5}. A Hamming distance of one separates these two $5$-grams, and many other assembly instructions also have fixed-length equivalent counterparts. 
Situations like this one are especially likely to occur in register renaming and changes of the function or data addresses that are prevalent even in non-adversarial recompilation. 
If we could learn to detect the regular expression \texttt{55 ? D2 89 E5}, we would have a feature that is robust to these changes.

Thus far, automated methods for constructing regular expressions for use in detecting malware have been hampered by the technical limitations of existing algorithms. 
For example, regular language induction methods, having been designed for problems in natural language rather than assembly,  
can not be applied directly due to the much larger sequence lengths and sample sizes. 
Prior approaches have also relied on genetic algorithms and similar search strategies that are not scalable ~\citep{10.1145/322326.322334,10.5555/2832581.2832710}.
The open challenge is in finding regular expressions in a computationally efficient manner that can readily scale to the data necessary to train modern-day malware systems.

In this work, we propose a new method of mining frequently occurring patterns based on the Hamming similarity between fixed-length token sequences. We allow wildcard tokens; a sequence with $\mathcal{K}$ wildcard tokens can match a set of $n$-grams within a Hamming distance of $\mathcal{K}$. Thus, we term these features \textit{hamm-grams}. 
As a simple example, the hamm-gram \texttt{AB?D} matches both \texttt{ABCD} and \texttt{ABFD}, but does not match \texttt{ABDC}.
By restraining ourselves to the space of hamming similarity, a set of hamm-grams forms a regular language. This guarantees that a search for matches can be done efficiently in linear time with respect to the input sequence length.

The rest of this article is organized as follows. First, we briefly discuss the related work in regular language induction, regular expressions, and $n$-grams in \autoref{sec:related_work}. Next, we design an efficient algorithm for hamm-gram extraction in \autoref{sec:method}. 
Finally, to demonstrate how our approach provides benefits for both malware family classification and malicious vs benign classification, we perform studies on the publicly available Drebin android malware family dataset and Ember 2018 Windows malware dataset in \autoref{sec:results}. 

\section{Related Work} \label{sec:related_work}

Regular expressions have been important building blocks of malware signature detection for decades, across file types and platforms 
\citep{Hahn2014,smutz2012malicious,VanBaar2014,Levchenko:2011:CTE:2006077.2006780}, 
and are still used by most modern anti-virus products\citep{Botacin2021a}. 
However, these regular expressions have thus far been constructed by human effort. Reverse-engineering a binary file to identify signatures can take days or weeks for an expert analyst \citep{Votipka2019,10.1145/3460120.3484759}, 
and so the acceleration and automation of this process is highly valuable. 

The field of regular language induction has benefited from significant study
\citep{10.5555/645515.658238,Coste1997,KUDO1988401,10.1145/322326.322334,10.5555/2832581.2832710}, 
but the methods that have been developed are not applicable malware detection. 
Those methods often require a fixed subset of languages and a corpus of \textit{known} positive and negative samples, and incur hard-to-quantify computational complexity and slow inference. 
In our case, the set of possible exemplars is all observed $n$-grams, on the order of $10^{12}$ unique substrings in large corpora, which is orders-of-magnitude beyond the scales ($10^{3}$ or less) for which regular language induction methods are feasible.

The method we propose uses algorithmic tools from regular languages to build deterministic finite automata (DFA) to efficiently search for matches to all stored hamm-grams.
Moreover, we restrict our regular expressions to the space of hamm-grams to create a computationally feasible approach. 
This extends the literal characters by a symbol that can match any single literal character (which we denote throughout by `?') and requires each expression to have a predetermined fixed length. 
Consequently, this subset of DFA can be implemented by a prefix tree (Trie) extended to allow child nodes to be wildcards.

$n$-grams have been popular features since the first works applying machine learning to the malware problem \citep{Kephart:1995:BID:1625855.1625983}, and have remained popular in both machine learning and security venues~\citep{Jang2011,Kolter2004,Drew2016,shareghi_show_2019} with significant effort to reduce the cost of $n$-gram extraction~\cite{Ceska2007,raff_zipf-gramming_2025,curtin_intermediate_2025,hashgram_2018,raff_hash_gram_parallel,gupta_living_2024}.
Prior attempts to build static features with more flexibility than $n$-grams have not become popular due to reduced accuracy. 
The study \citep{walenstein2007exploiting} proposed ``$n$-perms'', where a byte-sequence was sorted to a canonical order to reduce specificity. As an example, \texttt{CBA}, \texttt{BCA}, and \texttt{ABC} would all be considered equivalent $n$-perms because they are all identical to \texttt{ABC} after being sorted. 
Others have attempted to reduce $n$-gram specificity by first disassembling an executable to apply domain knowledge: e.g., \citep{Shabtai200916} proposed reducing instructions like 
\texttt{mov eax,42}
to just the instruction mnemonic 
\texttt{mov};
\citep{Masud2008} proposed reducing the operands to a reduced set of tokens denoting the type, for example, 
\texttt{mov eax,42}
would become \texttt{mov.register.constant} to denote that the first and second arguments were a register type and a constant type. However, subsequent works applying these techniques to larger industry-sized corpora have found that simple byte $n$-grams were more accurate \citep{Zak2017}. In contrast, our work develops a method that finds hamm-grams without any extra domain knowledge. Other challenges in malware detection, like measuring effect size and label confidence, are beyond the scope of this chapter~\cite{TirthCAMLIS,Joyce2022,joyce_claravy_2025,agtr}.

\section{Complex LSH Method and properties} \label{sec:method}

Creating hamm-grams from a candidate collection of $n$-grams first requires a way to bin $n$-grams of small mutual Hamming distance together.  
Otherwise, an algorithm would require a number of comparisons that scales at least quadratically in the total collection size. 
This binning will be achieved by developing a rolling hash function that is locality-sensitive by Hamming distance. The approach will be given in \autoref{sec:lsh_base}, and we show analytically that it corresponds to the Hamming distance in \autoref{sec:lsh_approx}. In \autoref{sec:lsh_random_experiment} we confirm experimentally that this rolling LSH accomplishes our goals.  
Once $n$-grams have been hashed such that similar $n$-grams collide, 
we employ 
an agglomerative clustering method
for mining hamm-grams from the buckets, and this is described in \autoref{sec:lsh_hamming_extraction}. 

\subsection{Description of Complex LSH} \label{sec:lsh_base}

Consider sequences of length $m$,
\begin{equation}
    [s_0, s_1, ..., s_{m-1}],
\end{equation}
where $s_i$ are drawn from a finite alphabet $\mathcal{A}$ of $T$ characters.
Let our \emph{tokenizer} $\mathcal{T}$, defined such that
\begin{equation}
    \mathcal{T}: \mathcal{A} \rightarrow \{ e^{i\psi_0}, e^{i\psi_1}, ..., e^{i\psi_T}\},
\end{equation}
be a bijective function 
that 
maps every element of the alphabet to a unique complex number with unit modulus. 
It follows that each sequence can be uniquely represented by a sequence of tokens:
\begin{equation}
    [\mathcal{T}(s_0), \mathcal{T}(s_1), ..., \mathcal{T}(s_{m-1})] \equiv [t_0, t_1, ..., t_{m-1}].
\end{equation}

Consider breaking each sequence into $n$-grams:%
\begin{equation}
    [t_0, t_1, ..., t_{n-1}], ..., [t_{m-n-1}, t_{m-n}, ..., t_{m-1}].
\end{equation}
We require a locality-sensitive hashing algorithm that is sensitive to the Hamming distance~\citep{hamming} between $n$-grams and can be computed in \textit{sub-linear} time in $n$. The algorithm we developed that fulfills these needs is described below. 

First we fill a vector $\boldsymbol{w}$ of size $n$ with powers of unit-modulus complex number $a$,
\begin{equation}
    \boldsymbol{w} \equiv [a^1, a^2, ..., a^n]
\end{equation}
where
\begin{equation}
    a \equiv e^{i \phi},
\end{equation}
and $\phi$ is sampled from $\rm{Uniform}[0, 2\pi]$. 
This is convenient for taking powers:
\begin{equation}
    a^r = e^{i r \phi},
\end{equation}
from which it follows that
\begin{equation}
    a^{-1} = e^{-i\phi}.%
\end{equation}
In the following, we use the notation $\boldsymbol{v}_{[i:j]}$ to denote the slice of a vector $\boldsymbol{v}$ beginning at $v_i$ and ending at $v_{j-1}$.
Consider the dot product between our weight vector $\boldsymbol{w}$ and our tokens, 
\begin{equation}
\label{eqn:curr_hash}
    h_p \equiv \boldsymbol{w} \cdot \boldsymbol{t}_{[j:j+n]} = a^1 t_j + a^2 t_{j+1} + \cdots + a^n t_{j+n-1}.
\end{equation}
Further, note that the dot product between the weight vector and the adjacent window is given by
\begin{equation}
    \boldsymbol{w} \cdot \boldsymbol{t}_{[j+1:j+n+1]} = a^1 t_{j+1} + a^2 t_{j+2} + \cdots + a^n t_{j+n}.
\end{equation}
Using that
\begin{equation}
    a^{-1}h_p - t_j = a^1 t_{j+1} + a^2 t_{j+2} + \cdots + a^{n-1}t_{j+n-1},
\end{equation}
it follows that
\begin{equation}
    h_u \equiv \boldsymbol{w} \cdot \boldsymbol{t}_{[j+1:j+n+1]} = a^{-1}h_p - t_j + a^{n}t_{j+n}.
\end{equation}
Thus, we've shown that the updated complex hash $h_u$ (the dot-product between the adjacent window and the weights) can be computed with complex multiplication and addition involving only the previous complex hash $h_p$, the excluded token $t_j$, and the new token $t_{j+n}$. This computation scales as $\mathcal{O}(1)$
with respect to the vector length $n$, and therefore will be efficient enough to be applied to every $n$-gram in each file.

\begin{figure}[!h]
\centering

\tikzset{every picture/.style={line width=0.75pt}} %

\begin{tikzpicture}[x=0.75pt,y=0.75pt,yscale=-1,xscale=1]
\draw    (150,23) -- (150,257) ;
\draw [shift={(150,260)}, rotate = 270] [fill={rgb, 255:red, 0; green, 0; blue, 0 }  ][line width=0.08]  [draw opacity=0] (8.93,-4.29) -- (0,0) -- (8.93,4.29) -- cycle    ;
\draw [shift={(150,20)}, rotate = 90] [fill={rgb, 255:red, 0; green, 0; blue, 0 }  ][line width=0.08]  [draw opacity=0] (8.93,-4.29) -- (0,0) -- (8.93,4.29) -- cycle    ;
\draw    (23,170) -- (267,170) ;
\draw [shift={(270,170)}, rotate = 180] [fill={rgb, 255:red, 0; green, 0; blue, 0 }  ][line width=0.08]  [draw opacity=0] (8.93,-4.29) -- (0,0) -- (8.93,4.29) -- cycle    ;
\draw [shift={(20,170)}, rotate = 0] [fill={rgb, 255:red, 0; green, 0; blue, 0 }  ][line width=0.08]  [draw opacity=0] (8.93,-4.29) -- (0,0) -- (8.93,4.29) -- cycle    ;
\draw [color={rgb, 255:red, 74; green, 144; blue, 226 }  ,draw opacity=1 ][fill={rgb, 255:red, 80; green, 227; blue, 194 }  ,fill opacity=1 ]   (150,170) -- (207.7,121.92) ;
\draw [shift={(210,120)}, rotate = 140.19] [fill={rgb, 255:red, 74; green, 144; blue, 226 }  ,fill opacity=1 ][line width=0.08]  [draw opacity=0] (10.72,-5.15) -- (0,0) -- (10.72,5.15) -- (7.12,0) -- cycle    ;
\draw [color={rgb, 255:red, 80; green, 227; blue, 194 }  ,draw opacity=1 ][fill={rgb, 255:red, 80; green, 227; blue, 194 }  ,fill opacity=1 ]   (210,120) -- (248.34,62.5) ;
\draw [shift={(250,60)}, rotate = 123.69] [fill={rgb, 255:red, 80; green, 227; blue, 194 }  ,fill opacity=1 ][line width=0.08]  [draw opacity=0] (10.72,-5.15) -- (0,0) -- (10.72,5.15) -- (7.12,0) -- cycle    ;
\draw [color={rgb, 255:red, 245; green, 166; blue, 35 }  ,draw opacity=1 ][fill={rgb, 255:red, 80; green, 227; blue, 194 }  ,fill opacity=1 ]   (210,120) -- (248.34,177.5) ;
\draw [shift={(250,180)}, rotate = 236.31] [fill={rgb, 255:red, 245; green, 166; blue, 35 }  ,fill opacity=1 ][line width=0.08]  [draw opacity=0] (10.72,-5.15) -- (0,0) -- (10.72,5.15) -- (7.12,0) -- cycle    ;
\draw [color={rgb, 255:red, 65; green, 117; blue, 5 }  ,draw opacity=1 ][fill={rgb, 255:red, 65; green, 117; blue, 5 }  ,fill opacity=1 ]   (250,60) -- (182.97,69.58) ;
\draw [shift={(180,70)}, rotate = 351.87] [fill={rgb, 255:red, 65; green, 117; blue, 5 }  ,fill opacity=1 ][line width=0.08]  [draw opacity=0] (10.72,-5.15) -- (0,0) -- (10.72,5.15) -- (7.12,0) -- cycle    ;
\draw [color={rgb, 255:red, 65; green, 117; blue, 5 }  ,draw opacity=1 ][fill={rgb, 255:red, 80; green, 227; blue, 194 }  ,fill opacity=1 ]   (250,180) -- (182.97,189.58) ;
\draw [shift={(180,190)}, rotate = 351.87] [fill={rgb, 255:red, 65; green, 117; blue, 5 }  ,fill opacity=1 ][line width=0.08]  [draw opacity=0] (10.72,-5.15) -- (0,0) -- (10.72,5.15) -- (7.12,0) -- cycle    ;
\draw [color={rgb, 255:red, 80; green, 227; blue, 194 }  ,draw opacity=1 ][fill={rgb, 255:red, 80; green, 227; blue, 194 }  ,fill opacity=1 ]   (130.82,237.12) -- (150,170) ;
\draw [shift={(130,240)}, rotate = 285.95] [fill={rgb, 255:red, 80; green, 227; blue, 194 }  ,fill opacity=1 ][line width=0.08]  [draw opacity=0] (10.72,-5.15) -- (0,0) -- (10.72,5.15) -- (7.12,0) -- cycle    ;
\draw [color={rgb, 255:red, 139; green, 87; blue, 42 }  ,draw opacity=1 ][fill={rgb, 255:red, 80; green, 227; blue, 194 }  ,fill opacity=1 ]   (130,240) -- (72.3,191.92) ;
\draw [shift={(70,190)}, rotate = 39.81] [fill={rgb, 255:red, 139; green, 87; blue, 42 }  ,fill opacity=1 ][line width=0.08]  [draw opacity=0] (10.72,-5.15) -- (0,0) -- (10.72,5.15) -- (7.12,0) -- cycle    ;
\draw [color={rgb, 255:red, 144; green, 19; blue, 254 }  ,draw opacity=1 ][fill={rgb, 255:red, 80; green, 227; blue, 194 }  ,fill opacity=1 ]   (70,190) -- (60.49,132.96) ;
\draw [shift={(60,130)}, rotate = 80.54] [fill={rgb, 255:red, 144; green, 19; blue, 254 }  ,fill opacity=1 ][line width=0.08]  [draw opacity=0] (10.72,-5.15) -- (0,0) -- (10.72,5.15) -- (7.12,0) -- cycle    ;
\draw  [draw opacity=0][fill={rgb, 255:red, 248; green, 231; blue, 28 }  ,fill opacity=0.1 ][dash pattern={on 4.5pt off 4.5pt}] (150,0) -- (300,0) -- (300,260) -- (150,260) -- cycle ;
\draw  [draw opacity=0][fill={rgb, 255:red, 126; green, 211; blue, 33 }  ,fill opacity=0.1 ][dash pattern={on 4.5pt off 4.5pt}] (0,0) -- (150,0) -- (150,260) -- (0,260) -- cycle ;

\draw (151,101.4) node [anchor=north west][inner sep=0.75pt]  [color={rgb, 255:red, 74; green, 144; blue, 226 }  ,opacity=1 ]  {$a^{1}\mathcal{T}( f)$};
\draw (242,71.4) node [anchor=north west][inner sep=0.75pt]  [color={rgb, 255:red, 80; green, 227; blue, 194 }  ,opacity=1 ]  {$a^{2}\mathcal{T}( a)$};
\draw (171,32.4) node [anchor=north west][inner sep=0.75pt]  [color={rgb, 255:red, 65; green, 117; blue, 5 }  ,opacity=1 ]  {$a^{3}\mathcal{T}( x)$};
\draw (201,192.4) node [anchor=north west][inner sep=0.75pt]  [color={rgb, 255:red, 65; green, 117; blue, 5 }  ,opacity=1 ]  {$a^{3}\mathcal{T}( x)$};
\draw (241,131.4) node [anchor=north west][inner sep=0.75pt]  [color={rgb, 255:red, 245; green, 166; blue, 35 }  ,opacity=1 ]  {$a^{2}\mathcal{T}( o)$};
\draw (92,181.4) node [anchor=north west][inner sep=0.75pt]  [color={rgb, 255:red, 80; green, 227; blue, 194 }  ,opacity=1 ]  {$a^{1}\mathcal{T}( a)$};
\draw (43,221.4) node [anchor=north west][inner sep=0.75pt]  [color={rgb, 255:red, 139; green, 87; blue, 42 }  ,opacity=1 ]  {$a^{2}\mathcal{T}( c)$};
\draw (71,131.4) node [anchor=north west][inner sep=0.75pt]  [color={rgb, 255:red, 144; green, 19; blue, 254 }  ,opacity=1 ]  {$a^{3}\mathcal{T}( e)$};
\draw (85.5,9.5) node   [align=left] {\begin{minipage}[lt]{88.57pt}\setlength\topsep{0pt}
\begin{center}
"0"-bit result
\end{center}

\end{minipage}};
\draw (214.5,10.5) node   [align=left] {\begin{minipage}[lt]{88.57pt}\setlength\topsep{0pt}
\begin{center}
"1"-bit result
\end{center}

\end{minipage}};

\end{tikzpicture}

\caption{An illustration of how the complex hashing function tends to put similar sequences into the same hash bucket, illustrated on $3$-grams of English letters: \texttt{fax} and  \texttt{fox} which get a ``1''-bit for their hash, and \texttt{ace} which receives a ``0'' bit for it's hash. Each $a^i \mathcal{T}(\cdot)$ applies a new rotation and fixed step size. The accumulation of such operations results in similarly shaped trajectories that differ in expectation as a function of how many steps in the sequence differ, i.e., the hamming distance. It is thus more likely for two sequences with low hamming distance to end up in the same half of the plane. E.g., fax and fox are both on the right hand size and receive a ``0'' bit in this example. }
\label{fig:clsh}
\end{figure}
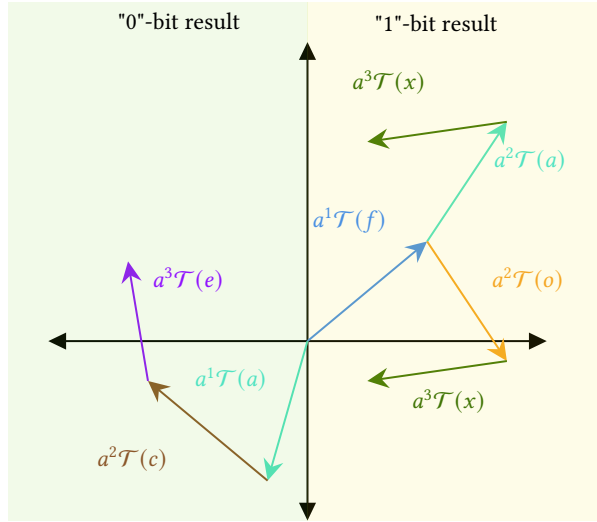

In Fig.~\ref{fig:clsh} we've provided an illustration that explains intuitively how the hash we have devised is locality-sensitive under the Hamming distance. As an example, we consider computing the hash on the $3$-grams of English letters, \texttt{fax}, \texttt{fox}, and \texttt{ace}.
Each of the three letters is assigned (randomly) a unit vector in the complex plane, which are denoted $\mathcal{T}(\text{f})$, $\mathcal{T}(\text{o})$, and so on for all letters considered. Now, according to equation ~\ref{eqn:curr_hash}, the complex value assigned to the $3$-gram \texttt{fox} is $a^1 \mathcal{T}(\text{f}) + a^2 \mathcal{T}(\text{o}) + a^3 \mathcal{T}(\text{x})$, and similarly for the others. Because $a$ and all of the complex values assigned to the letters are unit-length vectors, their products are also unit-length vectors. Therefore, in each sum we have unit-length vectors added tip-to-tail. 
Now, it should be clear why \texttt{fox} and \texttt{fax} are more likely to be mapped to complex values that lie closer together, and therefore more likely to fall within the same half-plane. This example also illustrates that although \texttt{fax} and \texttt{ace} both have the character \texttt{a}, because they occur at different positions in the sequence, 
they get multiplied by differing powers of $a$ and are mapped to different vectors in each sum.

Finally, once the complex hashes $h_i$ have been computed for all possible windows, a bit hash-code $b_i$ can be assigned to each according to the truth value of the following statement:
\begin{equation}
    b_i \equiv (\mathcal{R}\{h_i\} > 0).
\end{equation}
This entire process can be repeated $M$ times with independently and randomly sampled angles $\phi$ to define an $M$-bit hash code. 
In the following sections, we demonstrate both by analytic considerations and experiments that the hash code we have developed is locality-sensitive under the Hamming distance. 

\subsection{Semi-analytic approximation} \label{sec:lsh_approx}

In this section, we demonstrate that the hashing algorithm described above is locality-sensitive with respect to the Hamming distance (rather than the Jaccard distance~\citep{jaccard1912distribution}, for example). 
To do so, we derive certain approximate analytic results under the assumption of uniformly random sequences.
To that end, we consider again the expression for the dot product between a random weight vector $\boldsymbol{w}$ and a particular tokenized n-gram, 
\begin{equation}
\label{eqn:dot_prod}
    \boldsymbol{w} \cdot \boldsymbol{t} = e^{i[\phi + \psi_0]} + e^{i[2\phi + \psi_1]} +  \cdots + e^{i[n\phi + \psi_{n-1}]}.
\end{equation}
Since both the weight-vector phasors $e^{i j \phi}$ and the token phasors $e^{i \psi_{j}}$ are drawn from circularly symmetric distributions, it follows that their product $e^{i(j\phi)} e^{i\psi_{j-1}} = e^{i[(j\phi)+\psi_{j-1}]}$ is also a uniform distribution over the ring of unit modulus~\citep{tse2005fundamentals}.  
Therefore Eqn. \ref{eqn:dot_prod} is equivalent to a sum of the form 
\begin{equation}
    Z_n \equiv \sum_{j=1}^{n} e^{i\theta_j},
\end{equation}
where $\theta_j \sim \rm{Uniform}[0, 2\pi]$. 
This sum has familiar and useful properties~\citep{weisstein}, which we briefly review.
We note that the following results are derived in the approximation $n, k \gg 1$, and we use the notation $\langle \cdots \rangle$ for expectation values.

\begin{lemma}[Properties of \(Z_n\)]
The random walk \(Z_n = \sum_{j=1}^n e^{i\theta_j}\), where \(\theta_j \sim \mathrm{Uniform}[0, 2\pi]\), satisfies the following properties:
\[
\langle Z_n \rangle = 0, \quad \langle |Z_n|^2 \rangle = n, \quad \langle \mathcal{R}\{Z_n\}^2 \rangle = \langle \mathcal{I}\{Z_n\}^2 \rangle = n/2.
\]
Furthermore, for large \(n\), the distribution of \(x_n = \mathcal{R}\{Z_n\}\) is approximately
\[
x_n \sim \mathcal{N}(0, n/2).
\]
\end{lemma}

\begin{proof}

The distribution is circularly symmetric, and therefore 
\begin{equation}
    \langle Z_n \rangle = \langle \mathcal{R}\{Z_n\} \rangle = \langle \mathcal{I}\{Z_n\} \rangle = 0.
\end{equation}
Moreover, 
\begin{align}
    \langle |Z_n|^2 \rangle = \left\langle \sum^{n}_{i=1} \sum^{n}_{k=1} e^{i(\theta_j - \theta_k)} \right\rangle \\ 
    = n + \left\langle \sum^{n}_{j \neq k} e^{i(\theta_j-\theta_k)} \right\rangle = n
\end{align}
Since by definition,
\begin{equation}
    |Z_n|^2 = \mathcal{R}\{Z_n\}^2 + \mathcal{I}\{Z_n\}^2, 
\end{equation}
it follows by symmetry that 
\begin{equation}
    \langle \mathcal{R}\{Z_n\}^2 \rangle = \langle \mathcal{I}\{Z_n\}^2 \rangle = n/2.
\end{equation}

The sum $Z_n$ is equivalent to a random walk in the complex plane with unit-length steps. The resulting approximate probability distribution in the limit of large-n is usually attributed to Lord Rayleigh~\citep{Rayleigh:1905zps}. In the continuous limit, the problem is equivalent to a diffusion equation, which yields a probability density given by the multivariate Gaussian.   
Therefore, denoting $x_n \equiv \mathcal{R}\{Z_n\}$, we can approximate the probability density for $x_n$ by
\begin{equation}
    x_n \sim \mathcal{N}(0, n/2). 
\end{equation}

\end{proof}

\begin{lemma}[Probabilistic Relation for Hamming Distance]
Let \(x_c\) and \(\Delta x\) be the real parts of the sums of random walks for vectors differing by Hamming distance \(k\). The probability that \(|\Delta x| > |x_c|\) is given by
\[
p = \frac{1}{\pi} \tan^{-1}\left(\sqrt{\frac{\bar{k}}{1-\bar{k}}}\right), \quad \text{where } \bar{k} = k/n.
\]
\end{lemma}

\begin{proof}

Now consider the complex hashes (before we take the sign of the real part) of two vectors $\boldsymbol{t}$ and $\boldsymbol{t}'$ which differ by exactly $k$ tokens; that is, the two vectors are separated by Hamming distance equal to $k$. 
Let $t_c$ be the subvector of length $n-k$ which is common to both $\boldsymbol{t}$ and $\boldsymbol{t}'$, and $\boldsymbol{w}_c$ the corresponding subvector (according to index) of $\boldsymbol{w}$. Similarly, let $\boldsymbol{t}_d$ and $\boldsymbol{t}_d'$ be the subvectors of $k$ elements which are unique to $\boldsymbol{t}$ and $\boldsymbol{t}'$ respectively, and $\boldsymbol{w}_d$ the corresponding subvector (according to index) of $\boldsymbol{w}$. 

To more easily estimate the probability that their bit hash codes differ, we note that it is equivalent to the following probability:
Sample three random walks starting from the origin, one of length $(n-k)$, and two of length $k$. Suppose, without loss of generality, that the walk of length $(n-k)$ yields a positive real value $x_c$. What is the probability that exactly one of the two walks of length $k$ yields a negative real value $\Delta x$ with a magnitude larger than $x_c$? If $p$ denotes the probability that the first walk yields a negative real value with a magnitude larger than $x_c$, from the independence of the two random walks, it follows that $1-p$ gives the probability that the second walk does not. Therefore, altogether the probability we seek is given by $2p(1-p)$, with a factor of two counting the two possible situations (either the first or second walk may be the one that yields the negative real value). 

To solve for $p$, we note that 
\begin{equation}
    p = \frac{1}{2}\mathcal{P}(|\Delta x| > |x_c|),
\end{equation}
where $x_c \sim \mathcal{N}(0, (n-k)/2)$ and $\Delta x \sim \mathcal{N}(0, k/2)$. The factor of $1/2$ accounts for the fact that a larger magnitude alone is not sufficient, but it must also have opposite sign (two possibilities which are equal by symmetry). 
The probability distribution of the absolute value of a zero-mean normal random variable is given by the half-normal distribution with zero location parameter, 
therefore $|x_c| \sim {\rm HN}(0, (n-k)/2)$ and $|\Delta x| \sim {\rm HN}(0, k/2)$. 
The cumulative distribution function for the half-normal distribution with zero location parameter and scale parameter $\sigma$ is given by 
\begin{equation}
    \mathcal{P}(|\Delta x| < a) = \rm{erf}\left(\frac{a}{\sqrt{2}\sigma}\right). 
\end{equation}
Let 
\begin{equation}
    I \equiv \mathcal{P}(|\Delta x| < |x_c|) = \int_0^{\infty} da \mathcal{P}(|\Delta x| < a)\rho(|x_c|).
\end{equation}
We have the result that 
\begin{equation}
    I = 1 - \frac{2}{\pi} \rm{tan}^{-1} \left( \frac{\sqrt{k}}{\sqrt{n-k}} \right).
\end{equation}
It follows that 
\begin{equation}
    \mathcal{P}(|\Delta x| > |x_c|) = \frac{2}{\pi} \rm{tan}^{-1} \left( \frac{\sqrt{k}}{\sqrt{n-k}} \right),
\end{equation}
and therefore that 
\begin{equation}
    p = \frac{1}{\pi} \rm{tan}^{-1} \left( \frac{\sqrt{k}}{\sqrt{n-k}} \right) = \frac{1}{\pi} \rm{tan}^{-1} \left( \sqrt{\frac{\bar{k}}{1-\bar{k}}} \right),
\end{equation}
where we have defined $\bar{k} \equiv k/n$. 

\end{proof}

\begin{theorem}[Collision Probability]
The probability of a hash-code collision given Hamming distance \(k\) is
\[
\mathcal{P}(\mathrm{coll.} \mid H_d = k) = 1 - 2p(1-p),
\]
where \(p\) is as given in Lemma 2. For \(\bar{k} \to 1\), the asymptotic behavior is
\[
\mathcal{P}(\mathrm{coll.} \mid H_d = k) \approx \frac{1}{2} + \frac{2}{\pi^2} \frac{1-\bar{k}}{\bar{k}}.
\]
\end{theorem}

\begin{proof}

The probability decays with a long tail as the scaled Hamming distance increases. 
Specifically, consider the Taylor expansion of $\rm{tan}^{-1}(z)$ as $z \rightarrow \infty$:
\begin{equation}
    \rm{tan}^{-1}(z) = \frac{\pi}{2} - \frac{1}{z} + \text{(Higher Order Terms) }. 
\end{equation}

Therefore, denoting 
\begin{equation}
    z \equiv \sqrt{\frac{\bar{k}}{1-\bar{k}}}, 
\end{equation}
we have that 
\begin{equation}
    p \approx \frac{1}{2} - \frac{1}{\pi}\sqrt{ \frac{1-\bar{k}}{\bar{k}} }.
\end{equation}

Therefore, for values of $\bar{k}$ close to one, and 
keeping only leading order terms, 
\begin{equation}
    \mathcal{P}({\rm coll.} | H_d = k) \approx \frac{1}{2} + \frac{2}{\pi^2}\frac{1-\bar{k}}{\bar{k}}.
\end{equation}

\end{proof}

\begin{figure}[!h]
\centering
\begin{subfigure}{.5\linewidth}
  \centering
  \includegraphics[width=.95\linewidth]{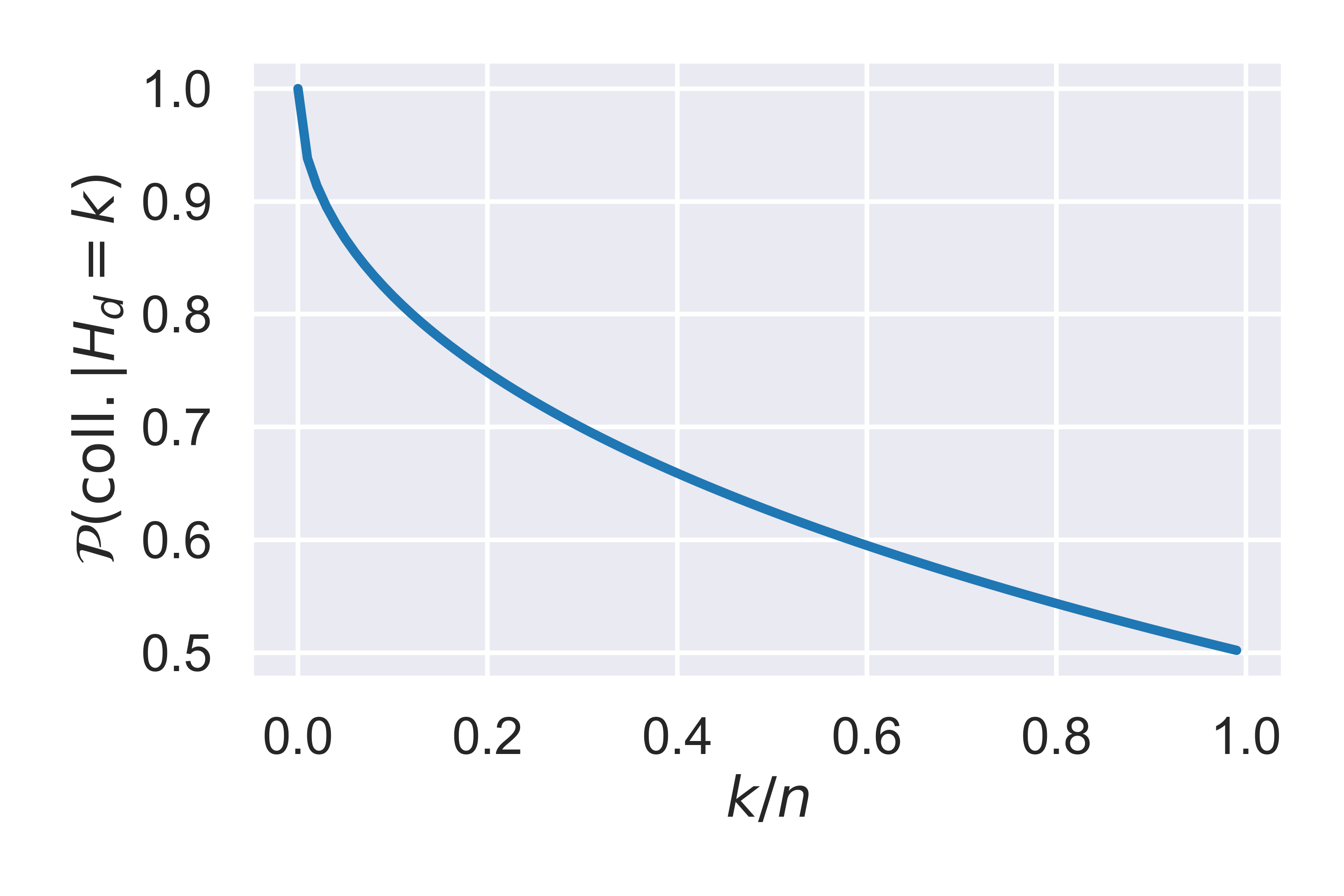}
\end{subfigure}%
\begin{subfigure}{.5\linewidth}
  \centering
  \includegraphics[width=.95\linewidth]{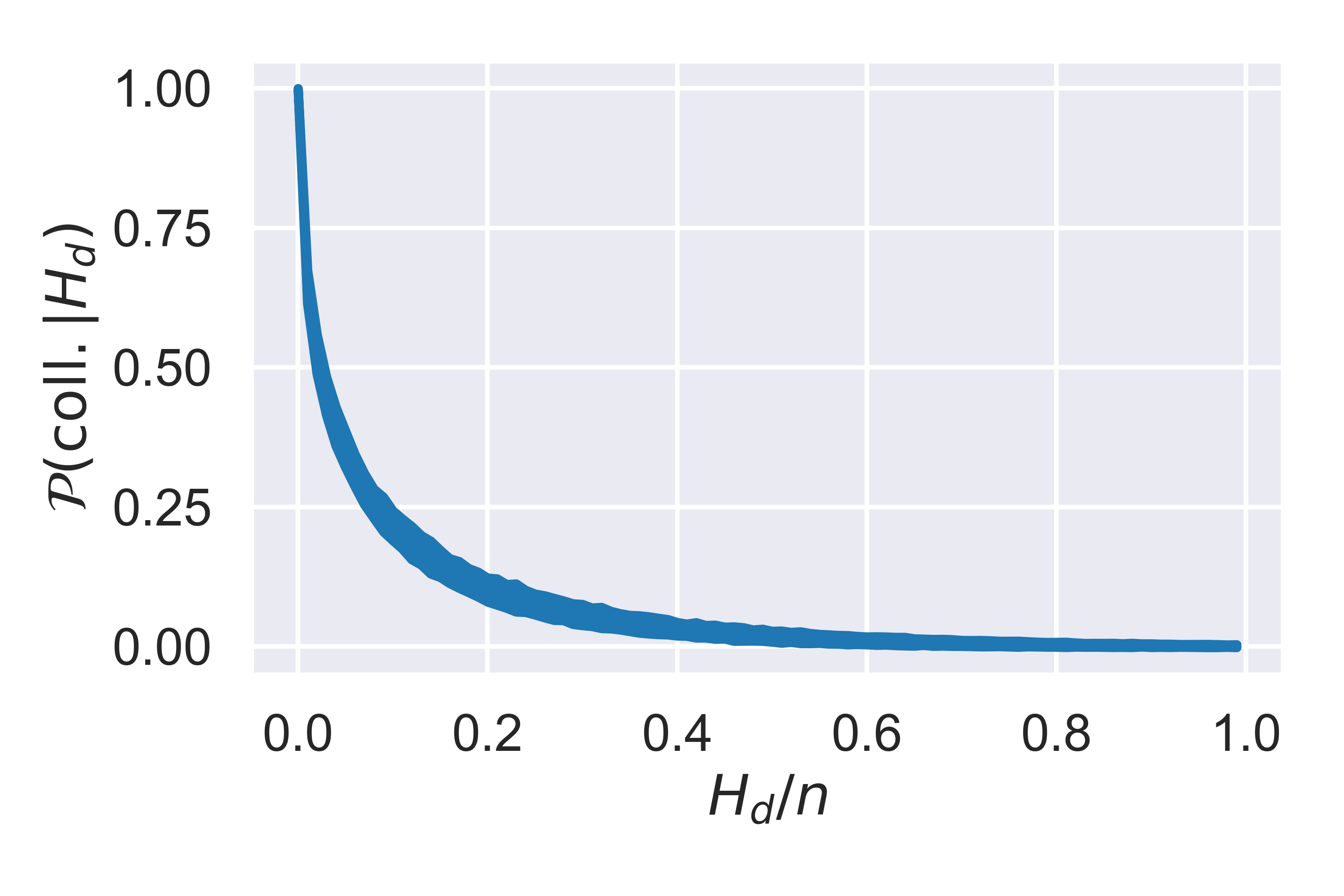}
\end{subfigure}
\caption{(left) Approximate probability of two $n$-grams having a single-bit hash code collision given that they differ by Hamming distance $k$. (right) 95\% credible interval for collision probability between pairs of $(n=100)$-grams as a function of normalized Hamming distance $H_d / n$. 
}
\label{fig:p_coll}
\end{figure}

We see that the derived approximations obey the expected asymptotic limits, with identical sequences being assigned the same bit code with probability one, while maximally different sequences have a purely random 50\% chance of being assigned the same bit (there are two possible bits). 
In practice, we use $M$-bit hash codes, and the probability of an $M$-bit collision will decay much more rapidly. This is observed in the experiment in the following section. 

\subsection{Experiments on random $n$-grams} \label{sec:lsh_random_experiment}

As an experimental test of our complex LSH scheme, we generate 100 independent random experiments. In each experiment, we first create 10 random weight vectors (giving our final bit hash-code space dimensionality $2^{10}$), create
1000 random $(n=100)$-grams by sampling an alphabet of 256 unique tokens, as well as creating Hamming neighbors by swapping exactly $H_d$ tokens with random alternates. The collision probabilities are shown in Fig.
~\ref{fig:p_coll} 
as a function of Hamming distance. The 95\% credible interval over random experiments is shown assuming a Gaussian distribution. 
We see that the probability of collision of all ten-bit hash codes decays rapidly as a function of (normalized) Hamming distance. 

\subsection{Extracting Hamm-Grams} \label{sec:lsh_hamming_extraction}

We have now established the algorithm by which we can run a rolling hash function over a sequence of bytes such that byte sequences with small Hamming distance are more likely to collide into the same hash bucket. Using this new hash, we extend the KiloGrams approach to extract commonly occurring hamm-grams.

\begin{algorithm}[!h]
\caption{Hamm-Gram Extraction}
\label{algo:hamming_gram}
\begin{algorithmic}[1]
\Require rolling hash function $h(\cdot)$,  corpus of  $\mathcal{C}$ documents, and desired number of frequent hamm-grams $k$

\State T $\gets $ new integer array %
\For{ all documents $x \in \mathcal{C}$} \Comment{First pass: find top-$k$ hashes}
    \For{$n$-gram $g \in x$} 
        \State $q' \gets h(g)$ %
        \State $T[q'] \gets T[q'] + 1$ 
    \EndFor
\EndFor
\State $T_k \gets \text{QuickSelect}(T, k)$ \Comment{top-$k$ hashes}

\State $S_H \gets $ $k$ different buckets each corresponding to a top-$k$ hash. 
\For{ all documents $x \in \mathcal{C}$} \Comment{Second pass: save $n$-grams in top-$k$ buckets}
    \For{$n$-gram $g \in x$} 
        \State $q' \gets h(g) $ %
        \If{$q' \in T_k$}
            \State Insert $g$ into the bucket corresponding to $q'$
        \EndIf
    \EndFor
\EndFor

\For{Each bucket of $n$-grams $S \in S_H$} \Comment{Cluster $n$-grams within each bucket}
    \While{$n$-grams remain in $S$ which can be merged}
        \State Apply agglomerative clustering on $n$-grams in $S$
    \EndWhile
\EndFor

\State Build a prefix tree of all saved hamm-grams
\State $F$ $\gets$ new map of final hamm-gram counts

\For{ all documents $x \in \mathcal{C}$} \Comment{Third pass: count matches and rank hamm-grams}
    \For{each character $c_i$ in document $x$}
        \If{a match is found}
            \State Record the match
            \State Update $F$ with the count of matches
        \EndIf
    \EndFor
\EndFor

\State \Return top-$k$ hamm-grams from $F$ ranked by the number of matches
\end{algorithmic}
\end{algorithm}

The KiloGrams algorithm \citep{raff2019kilograms} used a rolling hash function to find the top-$k$ hashes by frequency of n-grams, ignoring that collisions would occur. A second pass over the data would then use the Space Savings data structure to find the true top-$k$ $n$-grams by ignoring any $n$-gram that did not have a top-$k$ hash. Because the vast majority of $n$-grams are singletons, and the true top-$k$ are selected with high probability, it allowed them to find the true top-$k$ $n$-grams with high probability for almost arbitrary values of $n$. 

Selecting the top-$k$ hamm-grams is non-trivial, as finding a minimum set of hamm-grams ($n$-grams with wildcard tokens) is computationally intractable. 
We instead devise a procedure to build possible hamm-grams after the second pass, followed by a third pass which performs matching against the training corpus, sorts regular expressions by number of matches, and keeps only the top-$k$.

Our algorithm is detailed in Algorithm \autoref{algo:hamming_gram}, where the first two loops over all documents in a corpus $\mathcal{C}$ run a modified KiloGram algorithm, 
collecting the most frequent hashes, 
and then saving all $n$-grams that reside in the top-$k$ buckets. 
Agglomerative clustering with a fixed stopping condition is used to create hamm-grams from the n-grams within each bucket. 
This can be implemented in sub-cubic time (in the number of n-grams within the bucket) by an implementation using priority queues~\citep{manning2008introduction}. 
Two hamm-grams are candidates for `merging' if wildcards can be placed in a new expression until it matches both expressions without exceeding the predetermined budget of wildcards. 
Among all candidates, the two hamm-grams of the smallest Hamming distance are merged. 
This process is iterated until there are no mergeable hamm-grams left. 
Additionally, because the Hamming distance is defined as an integer, any pair found to have a Hamming distance of zero can be immediately merged without calculating the remaining distances. This results in runtime in practice significantly better than worst-case scaling.

\section{Hamm-Gram Results on Malware} \label{sec:results}

\subsection{Results on Drebin Android Malware Family Task}

We demonstrate and compare the performance of the top-$k$ hamm-grams and top-$k$ n-grams on the Drebin \citep{arp2014drebin} android malware family classification dataset. 
We have first taken a subset of the twenty malware classes in the dataset which have at least forty samples, leaving a dataset of twenty classes totaling 4661 samples.
In each experiment performed, we kept either the top-($k \mathord{=} 5e2$) or top-($k \mathord{=} 5e3$) grams. 
We begin by examining the learned $n$- and hamm-gram features, measures of information they carry regarding the malware family, and how well these features cluster into groups. Then we compare the performance of Logistic Regression models trained and these features.

\begin{figure}[!h]
\centering
\begin{subfigure}{.5\linewidth}
  \centering
  \includegraphics[width=.95\linewidth]{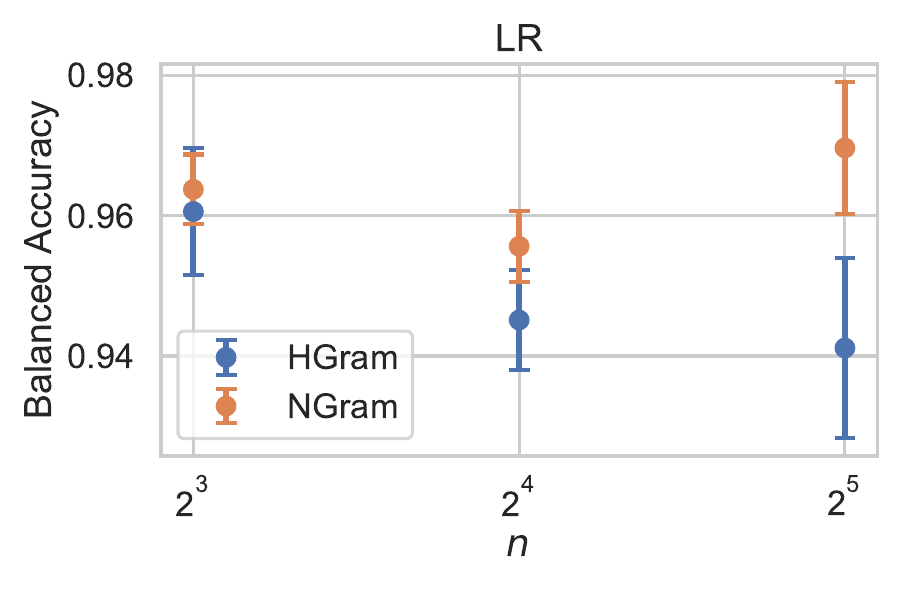}
\end{subfigure}%
\begin{subfigure}{.5\linewidth}
  \centering
  \includegraphics[width=.95\linewidth]{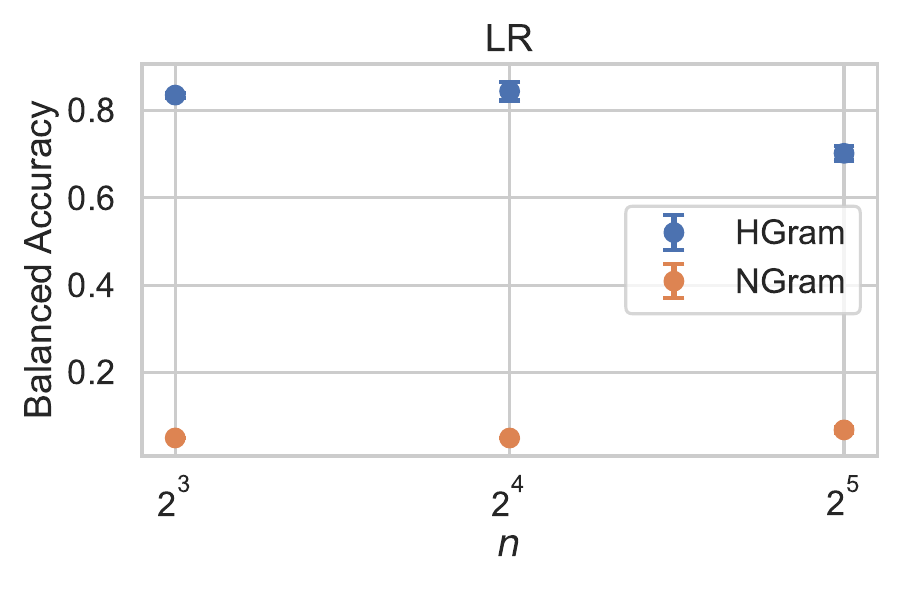}
\end{subfigure}
\caption{Balanced accuracy scores of Logistic Regression classifiers trained on the top-$(k \mathord{=} 5e3)$ (left) or top-$(k \mathord{=} 5e2)$ (right) features, as a function of gram size. The standard deviation over ten random experiments is shown as vertical bars. 
}
\label{fig:drebin_LR}
\end{figure}

Finally, we trained and tested Logistic Regression classifiers on random training/testing splits of the $n$-gram and hamm-gram features for various gram sizes $n$, for ten random experiments each. These are shown in 
Fig.~\ref{fig:drebin_LR}.
In the first case, we keep the top-$(k \mathord{=} 5e3)$ features. In this case, we see both $n$-gram and hamm-gram features achieving balanced accuracies of greater than 95\%. 
These accuracies are considerable given the challenge of a few examples per malware family.
In the second case, we keep only the top-$(k \mathord{=} 5e2)$ features. In this case, we see that for all gram sizes ($n$), hamm-grams significantly outperform ordinary $n$-grams. For $(n \mathord{=} 8)$-grams, the balanced accuracy is barely better than random guessing, while for size 8 hamm-grams the balanced accuracy is already better than 80 \%, significantly better, and large considering the difficulty of learning to classify twenty families of malware with as low as 40 examples of each. 
Thus we see that hamm-gram features can provide significant advantages in cases where model computations are considerably memory- and/or latency-constrained. 

\begin{figure}[!h]
\centering
 \includegraphics[width=0.75\columnwidth]{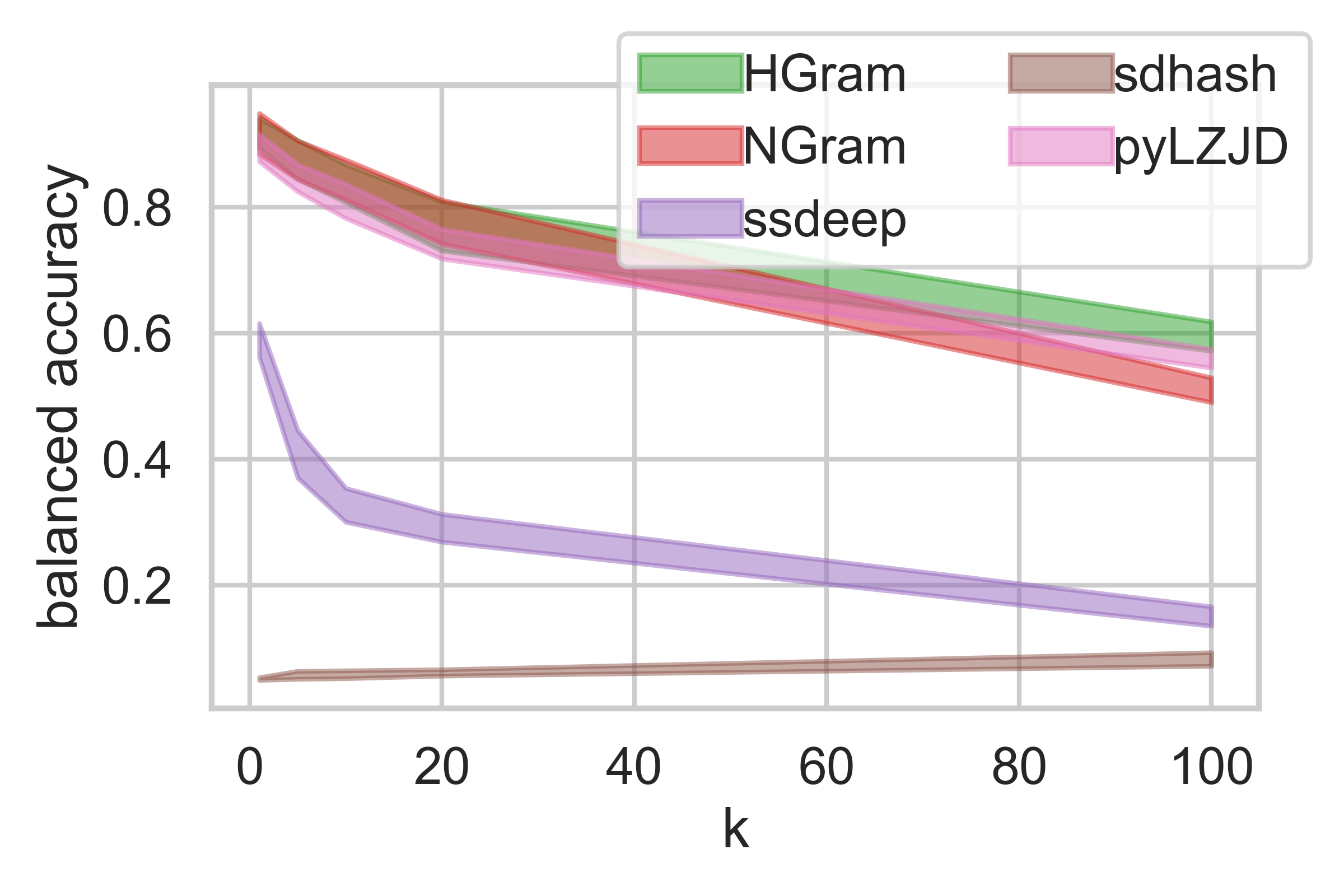}
\caption{Comparison between hamm-gram features and three additional hash baselines, ssdeep, sdhash, and LZJD.
}
\label{fig:additional_knn_baselines}
\end{figure}

Additionally, in Fig.~\ref{fig:additional_knn_baselines} we compare against three similarity hashing methods: ssdeep~\citep{kornblum2006identifying}, sdhash~\citep{roussev2010data},
and LZJD~\citep{raff2018lempel}. 
ssdeep is a fuzzy hash that runs hash functions on fixed-size segments of the data.
sdhash is a fuzzy hash that uses bloom filters and can generate similarity scores. 
LZJD creates Lempel-Ziv dictionaries and measures file similarity via Jaccard distance between their vocabularies.
Each of these methods is intended to work when some small number of bytes have been altered and for malware data. However, they are harder to deploy in many situations because the k-NN strategy requires storing all data and several orders of magnitude more compute than logistic regression. 
We also used the hamm-gram and $n$-gram feature vectors and k-NN classification under cosine distance. 
As can be seen, our hamm-gram features are significantly better than ssdeep and sdhash, two industry-standard tools, for similarity search.
Hamm-grams perform similarly to PyLZJD and $n$-grams for nearest-neighbor search on this dataset, 
but a hamm-gram LR classifier outperforms LZJD. This further shows the value of our approach in building malware family detectors that have more robustness to natural changes that occur as malware spreads and modifies, while being faster and lighter weight to deploy on endpoints in a network. 

\subsection{Examples of Learned Hamm-Gram Features for Drebin}

As empirically demonstrated, the hamm-gram approach significantly improves the accuracy when only a smaller number of features can be considered due to computational constraints. This was further expanded by exploring the top-$k$ features produced by each algorithm. Most notably, the normal $n$-gram approach suffers from extreme feature redundancy that reduces the effective information available to the underlying algorithm. In this section, we specifically examine certain explicit features learned on the Drebin corpus.

A prominent failure case of the standard $n$-gram features on the Drebin dataset are various alterations of the pattern \texttt{SHA1-Digest: ???}, each occupying capacity in the top-k features. On the other hand, the hamm-gram model with size 16 found this RegEx itself as a feature, freeing up space for other semantically different features. 
By coalescing these highly redundant $n$-grams, our hamm-gram approach enables more capacity in the final top-$k$ list for informative features that are both wild-card free (i.e., a normal $n$-gram that would not have been kept otherwise, because common and semantically redundant $n$-grams would have kept it out of the feature selection stage) or that actively use the wild-card capability.

Another example containing wild cards is variants of the byte sequence of \texttt{\textless FONT C????\textgreater.html} where the ``\texttt{?}'' indicates a wild card. 
In this case, it found a RegEx feature identifying a malicious family that uses fonts as a vector for malware. 
In each sample in the corpus, the font has a different name that always begins with a ``C'', and is four or five characters long.
This pattern was identified by our hamm-gram approach. Normal $n$-grams do not find this feature, since the created font names are each unique. 

\subsection{Results on Ember 2018 Malware Classification Task}

\begin{table}
\caption{Standardized Partial AUCs for different features sets and sizes, on EMBER 2018 dataset, for liblinear model with $k$=1e4.}
\label{tab:ember_k_1e4_model_liblinear}
\begin{tabular}{lccc}
\toprule
feature & $n$ & AUC(FPR $< 0.001$) & AUC(FPR $< 0.01$) \\
\midrule
$n$-gram & 8 & 0.580 & 0.737 \\
$n$-gram & 16 & 0.557 & 0.703 \\
$h$-gram & 8 & \textbf{0.605} & \textbf{0.770} \\
$h$-gram & 16 & 0.585 & 0.736 \\
\bottomrule
\end{tabular}
\end{table}

\begin{table}
\caption{Standardized Partial AUCs for different features sets and sizes, on EMBER 2018 dataset, for liblinear model with k=1e5.}
\label{tab:ember_k_1e5_model_liblinear}
\begin{tabular}{lccc}
\toprule
feature & $n$ & AUC(FPR $< 0.001$) & AUC(FPR $< 0.01$) \\
\midrule
$n$-gram & 8 & 0.684 & 0.857 \\
$n$-gram & 16 & 0.600 & 0.785 \\
$h$-gram & 8 & \textbf{0.746} & \textbf{0.893} \\
$h$-gram & 16 & 0.694 & 0.852 \\
\bottomrule
\end{tabular}
\end{table}

In this section, we report a comparison between Logistic Regression models trained and tested on the EMBER 2018 dataset~\citep{anderson2018ember} using either the top-$k$ $n$-gram features or top-$k$ hamm-gram features. The EMBER 2018 dataset consists of Windows Portable Executable (PE) files, $600, 000$ training samples, and $200, 000$ testing samples, split evenly between malicious and benign samples.

In Table~\ref{tab:ember_k_1e4_model_liblinear} are shown the standardized partial AUC scores for logistic regression models using either the top-$k$ $n$-gram features or top-$k$ hamm-gram features, for $k=10^4$. 
We show the standardized partial AUC because high false positive rates are unacceptable in deployed malware detection scenarios. 
These models were trained using LIBLINEAR with a Lasso penalty hyperparameter which was optimized on validation data by Optuna. We see that the model with hamm-gram features of length $8$ (and with a wildcard budget of $2$), is the best-performing model. Additionally, the model with hamm-gram features of length $16$ (and with a wildcard budget of $4$) performs similarly to the best-performing $n$-gram model. 

Similarly, in Table~\ref{tab:ember_k_1e5_model_liblinear} are shown the standardized partial AUC scores for logistic regression models using either the top-$k$ $n$-gram features or top-$k$ hamm-gram features, for $k=10^5$. 
Again, we see that the model with hamm-gram features of size $8$ (with a wildcard budget of $2$), has the strongest performance, and the model with hamm-gram features of length $16$ (and with a wildcard budget of $4$) performs similarly to the best performing $n$-gram model.

\subsection{Results on PDF Classification Task}

\begin{table}
\caption{Partial AUCs for different features sets and sizes, on PDF dataset, for liblinear model with k=1e4.}
\label{tab:pdf_k_1e4_model_liblinear}
\begin{tabular}{lccc}
\toprule
feature & $n$ & AUC(FPR $< 0.001$) & AUC(FPR $< 0.01$) \\
\midrule
$n$-gram & 8 & 0.957 & 0.991 \\
$n$-gram & 16 & 0.944 & 0.985 \\
$h$-gram & 8 & \textbf{0.985} & \textbf{0.994} \\
$h$-gram & 16 & 0.978 & 0.991 \\
\bottomrule
\end{tabular}
\end{table}

In this section, we report comparisons on a dataset of 100K train and 10k test malicious and benign PDFs scraped from VirusTotal. 
The comparison in Table~\ref{tab:pdf_k_1e4_model_liblinear} yields similar conclusions, with hamm-gram features of size 8 performing the best on the task, better than $n$-gram features of size 8 or 16. Additionally, the hamm-grams of size 16 perform as well as the best $n$-gram model ($n$=8). 

\section{Conclusion} \label{sec:conclsion}

In this manuscript, we have proposed a novel algorithm for mining \emph{hamm-grams}, common regular expressions of bytes, utilizing a newly developed Locality Sensitive Hash. 
We have demonstrated that hamm-grams provide significant increases in model performance 
on an Android malware classification task 
and on large malware detection tasks for Windows Portable Executables and PDFs. 
Finally, we have demonstrated examples of mined hamm-gram features where allowing wildcards was essential to discovering a meaningful underlying concept in the data.

\bibliographystyle{ACM-Reference-Format}
\bibliography{biblio}

\end{document}